\documentclass{llncs}

\usepackage{graphicx}
\usepackage{pifont}
\usepackage{verbatim}
\usepackage{amsmath,amsfonts}
\usepackage{enumerate,url}
\usepackage{paralist}
\usepackage{subfig}
\usepackage{wrapfig}
\usepackage{sidecap}

\begin{document}

\let\doendproof\endproof
\renewcommand\endproof{~\hfill\qed\doendproof}

\spnewtheorem{myremark}{Remark}{\bfseries}{\itshape}
\spnewtheorem{obs}{Observation}{\bf}{\it}
\newcommand{\down}[1]{\raisebox{-.4ex}{#1}}

\makeatletter
    \renewcommand*{\@fnsymbol}[1]{\ensuremath{\ifcase#1\or *\or \dagger\or \S\or
       \mathsection\or \mathparagraph\or \|\or **\or \dagger\dagger
       \or \ddagger\ddagger \else\@ctrerr\fi}}
\makeatother

\newcommand*\samethanks[1][\value{footnote}]{\footnotemark[#1]}

\title{A \boldmath$(7/2)$-Approximation Algorithm for Guarding Orthogonal Art
Galleries with Sliding Cameras}

\author{
Stephane Durocher$^1$\thanks{{Work of the author is supported in part by the Natural Sciences and Engineering Research Council of Canada
(NSERC).}} \and 	
Omrit Filtser$^2$\and
Robert Fraser$^1$\and\\
Ali Mehrabi$^3$\thanks{{Work of the author is supported by the Netherlands' Organization for Scientific Research (NWO)}}\and
Saeed Mehrabi$^1$\thanks{{Work of the author is supported in part by a University of Manitoba Graduate Fellowship (UMGF).}}
}

\institute{
$^1$Department of Computer Science, University of Manitoba, Canada.\\
$^2$Department of Computer Science, Ben-Gurion University of the Negev, Israel.\\
$^3$Department of Mathematics and Computer Science, Eindhoven University of Technology, Netherlands.\\
\email{durocher@cs.umanitoba.ca, omritna@cs.bgu.ac.il,  fraser@cs.umanitoba.ca, amehrabi@win.tue.nl, mehrabi@cs.umanitoba.ca}
}


\newcommand{\keywords}[1]{\par\addvspace\baselineskip
\noindent\keywordname\enspace\ignorespaces#1}
\newtheorem{observation}{Observation}{\bfseries}{\itshape}%

\pagenumbering{arabic}
\pagestyle{plain}

\maketitle

\begin{abstract}
Consider a sliding camera that travels back and forth along
an orthogonal line segment $s$ inside an orthogonal polygon $P$ with $n$ vertices. The camera can see a point $p$ inside
$P$ if and only if there exists a line segment containing $p$ that crosses $s$ at a right angle and 
is completely contained in $P$. In the minimum sliding cameras (MSC)
problem, the objective is to guard $P$ with the minimum number of sliding cameras.
In this paper, we give an $O(n^{5/2})$-time $(7/2)$-approximation algorithm to the
MSC problem on any simple orthogonal polygon with $n$ vertices,
answering a question posed by Katz and Morgenstern (2011).
To the best of our knowledge, this is the first constant-factor approximation algorithm for this problem.
\end{abstract}

\section{Introduction}
\label{sec:introduction}
In the classical art gallery problem, we are given
a polygon and the objective is to cover the polygon with the union of visibility
regions of a set of points guards while minimizing the number of guards. The problem was introduced by Klee
in 1973~\cite{orourke1987}. Two years later, Chv\'{a}tal~\cite{chvatal1975} showed that $\lfloor n/3\rfloor$
point guards are always sufficient and sometimes necessary to guard the polygon.
Since then, the problem and its many variants have been studied extensively for
different types of polygons (e.g., orthogonal polygons~\cite{dietmar1995} and polyominoes~\cite{biedl2012}), different types of
guards (e.g., points and line segments) and different visibility types.

Recently, Katz and Morgenstern~\cite{katz2011} introduced a variant of the art gallery problem
in which \emph{sliding cameras} are used to guard an
orthogonal polygon. Let $P$ be an orthogonal polygon with $n$ vertices. A sliding camera travels back and forth along an
orthogonal line segment $s$ inside $P$. The camera can
see a point $p\in P$ if the there is a point $q\in s$ such that
$pq$ is a line segment normal to $s$ that is completely inside $P$. In the minimum
sliding cameras ({\bf MSC}) problem, the objective is to guard $P$ using a
the minimum number of sliding cameras.

In this paper, we give an $O(n^{5/2})$-time $(7/2)$-approximation algorithm to the minimum sliding cameras (MSC)
problem on any simple orthogonal polygon.
To do this, we introduce the \emph{minimum guarded sliding cameras} ({\bf MGSC})
\emph{problem}. In the MGSC problem, the objective is to guard $P$ using a set of minimum
cardinality of guarded sliding cameras.
A sliding camera $s$ is guarded by a sliding camera $s'$ if every point on $s$ is seen by some point on $s'$. Note that $s$ and $s'$ could be perpendicular in which case $s'$ and $s$ mutually guard each other if and only if they intersect. If $s$ and $s'$ mutually guard each other and have the same orientation (e.g., both are horizontal), then the visibility region of $s$ is identical to that of $s'$. Consequently, when minimizing the number of sliding cameras in the MGSC problem, it suffices to consider solutions in which each horizontal sliding camera is guarded by a vertical sliding camera and vice-versa.
We first establish a connection between the MGSC problem and a related guarding problem on grids.

A grid $D$ is a connected union of vertical and horizontal line segments;
each maximal line segment in the grid is called a \emph{grid segment}. A \emph{point guard} $x$
in grid $D$ is a point that sees a point $y$ in the grid if the
line segment $xy\subset D$. Moreover, a \emph{sliding camera}
$p\in D$ is a point guard that moves along a grid
segment $s\in D$. The camera $p$ can see a point $q$ on the grid if and
only if there exists a point $p'\in s$ such that the line segment
$\overline{p'q}\subset D$; that is, point $q$ is seen by camera $p$ if
either $q$ is located on $s$ or $q$ belongs to a grid segment that
intersects $s$. Note that sliding cameras are called \emph{mobile guards} in
grid guarding problems~\cite{kosowski2004,adrian2006}. A \emph{guarded} set of point guards and a guarded set of sliding cameras on grids
are defined analogously to a guarded set of sliding cameras in polygons. A simple grid is defined as follows:

\begin{definition}[Kosowski et al.~\cite{adrian2006}]
A grid is \emph{simple} if \begin{inparaenum}[(i)]\item the
endpoints of all of its segments lie on the outer face of the planar
subdivision induced by the grid, and
\item there exists an $\epsilon>0$ such that every grid segment can
be extended by $\epsilon$ in both directions such that its new
endpoints are still on the outer face. \end{inparaenum}
\end{definition}

Throughout the paper, we denote a simple orthogonal polygon by $P$; note that the
polygon $P$ is a closed region.
The rest of the paper is organized as follows. Section~\ref{sec:relatedWork} presents
related work. In Section~\ref{sec:3approxAlgorithm}, we give our $(7/2)$-approximation
algorithm to the MSC problem and we conclude the paper in Section~\ref{sec:conclusion}.

\section{Related Work and Definitions}
\label{sec:relatedWork}
The minimum sliding cameras problem was introduced by Katz and Morgenstern~\cite{katz2011}.
They first considered a restricted version of the problem in which only vertical cameras are allowed; by reducing
the problem to the minimum clique cover problem on chordal graphs, they solved the problem exactly
in polynomial time. For the generalized case, where both vertical and horizontal cameras are allowed,
they gave a $2$-approximation algorithm for the MSC problem under the assumption that the polygon $P$
is $x$-monotone. 
Durocher and Mehrabi~\cite{durocher2013} showed that the MSC
problem is \textsc{NP}-hard when the polygon $P$ is allowed to have holes. They also gave an exact algorithm
that solves in polynomial time a variant of the MSC problem in which the objective is to minimize the sum of the lengths of line
segments along which cameras travel.

The guard problem on grids was first formulated by Ntafos~\cite{ntafos1986}.
He proved that a set of (stationary) point guards of minimum cardinality covering a grid of
$n$ grid segments has $n−m$ guards, where $m$ is the size of the maximum
matching in the intersection graph of the grid that can be found in $O(n^{5/2})$
time. Malafijeski and Zylinski~\cite{malafijeski2005} showed that the problem of
finding a minimum-cardinality set of guarded point guards for a grid is \textsc{NP}-hard.
Katz et al.~\cite{katz2005} showed that the problem of finding a minimum number
of sliding cameras covering a grid is \textsc{NP}-hard. Moreover,
Kosowski et al.~\cite{kosowski2004} proved that the problem of finding
the minimum number of guarded sliding cameras covering a grid (we call
this problem the {\bf MMGG} problem) is \textsc{NP}-hard.
Due to these hardness results, Kosowski et al.~\cite{adrian2006} studied
the MMGG problem on some restricted classes of grids. In particular, they show the following result
on simple grids:
\begin{theorem}[Kosowski et al.~\cite{adrian2006}.]
\label{thm:nSquarForGridProblem}
There exists an $O(n^2)$-time algorithm for solving the MMGG problem on
simple grids, where $n$ is the number of grid segments.
\end{theorem}


Throughout the paper, we denote optimal solutions for the MSC problem and the MGSC problem on $P$ by
$OPT_P$ and $OPT_{GP}$, respectively.
We denote the set of reflex vertices of $P$
by $V(P)$ and let $H_u$ and $V_u$ be the maximum-length
horizontal and vertical line segments, respectively, inside $P$ through a vertex
$u\in V(P)$. Let $L(P)=\{H_u\mid u\in V(P)\}\cup\{V_u\mid u\in V(P)\}$.
Let $L$ and $L'$ be two
orthogonal line segments (with respect to $P$ not necessarily each other) inside $P$; the \emph{visibility region} of $L$ is the union
of the points in $P$ that are seen by the sliding camera that travels along $L$. Moreover,
we say that $L$ dominates $L'$ if the visibility
region of $L'$ is a subset of that of $L$.

\section{A \boldmath$(7/2)$-Approximation Algorithm for the MSC Problem}
\label{sec:3approxAlgorithm}%
In this section, we present an $O(n^{5/2})$-time $(7/2)$-approximation algorithm for the MSC problem.

\subsection{Relating the MGSC and MMGG problems}
Consider an optimal solution $X$ for the MSC problem and let
$X'$ be the set of line segments obtained by taking two instances every line segment in $X$. We
observe that $X'$ is a feasible solution for the MGSC problem and, therefore, we have
the following observation.
\begin{observation}
\label{obs:mgscAndMSC}
An optimal solution for the MGSC problem on $P$ is a $2$-approximation to an
optimal solution for the MSC problem on $P$.
\end{observation}

We first consider how to apply a solution for the MGSC problem to the MSC problem.
For the MGSC problem, the idea is to reduce the MGSC problem to the MMGG problem.
Given any simple orthogonal polygon $P$, 
we construct a grid $G_P$ associated with $P$ as follows: initially, let $G_P$ be the
set of all line segments in $L(P)$. Now, for any pair of reflex vertices
$u$ and $v$ where $H_u$ dominates $H_v$ (resp., $V_u$ dominates $V_v$) in $P$, we remove
$H_v$ (resp., $V_v$) from $G_P$; if two segments mutually dominate each other, remove
one of the two arbitrarily. Let $T_G$ be the set of remaining grid segments in $G_P$. Observe that $G_P$
can be constructed in $O(n^2)$ time, where $n$ is the number of vertices of $P$. We first show the following result:

\begin{lemma}
\label{lem:gIsSimple}%
Grid $G_P$ is a simple and connected grid.
\end{lemma}
\begin{proof}
It is straightforward from the construction of $G_P$ that both endpoints of each grid
segment in $T_G$ lie on the boundary of $P$; this means that the endpoints of every grid
segment in $T_G$ lie on the outer face of $G_P$ and, therefore, $G_P$ is simple. To show that $G_P$
is connected, we first observe that the grid induced by the line segments in $L(P)$ is connected.
We now need to show that the grid remains connected after removing the set of grid
segments that are dominated by other grid segments. Let $s\in L(P)$ be a grid segment
that is removed from $L(P)$ (i.e., $s\notin T_G$). It is straightforward to see that the set of
grid segments that are intersected by $s$ are also intersected by $s'\in T_G$, where $s'$
is the grid segment that dominates $s$. Therefore, grid $G_P$ is connected.
\end{proof}
\begin{figure}[t]
\centering%
\includegraphics[width=0.85\textwidth]{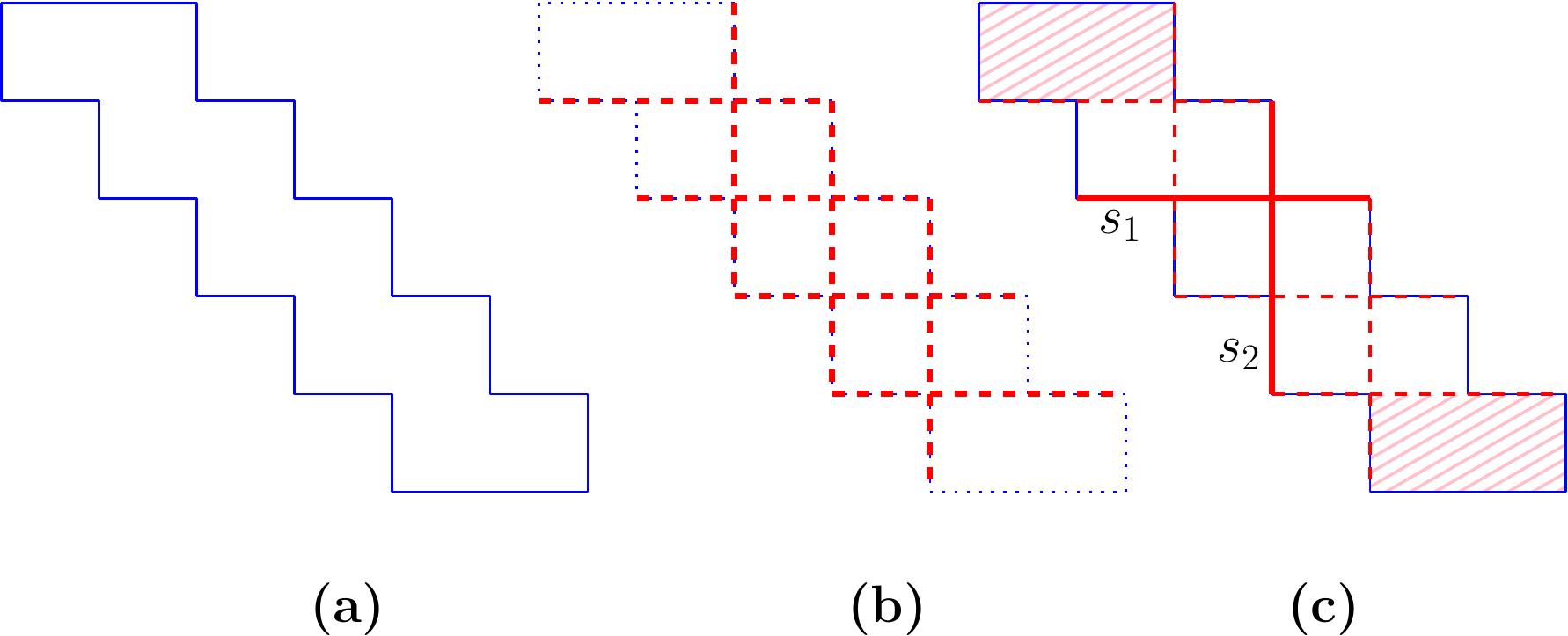}
\caption{(a) A simple orthogonal polygon $P$.
(b) Grid $G_P$ with set $T_G$ of grid segments shown in red. (c) The set
$S=\{s_1, s_2\}$ (represented by solid red line segments) is an optimal solution for the MMGG problem
on $G_P$, but $s_1$ and $s_2$ cannot guard $P$ entirely; in particular, the hatched regions of $P$ are not guarded.}
\label{fig:notEquivalent}%
\end{figure}

The objective is to solve the MMGG problem on $G_P$ exactly
and to use the solution $S$, the set of guarded grid segments computed, as the solution to the MGSC problem.
However, $S$ is not always a feasible solution to the MGSC problem since some regions in $P$
might remain unguarded; see Figure~\ref{fig:notEquivalent}
for an example. In the following, we characterize the regions of $P$ that may remain unguarded by the line segments in $S$;
we call these the \emph{critical regions} of $P$.

Consider $S$ and choose any unguarded point $p$ inside $P$. Let $R_p$ be a maximal axis-aligned
rectangle contained in $P$ that
covers $p$ and is also not guarded by the line segments in $S$. We observe that \begin{inparaenum}[(i)]\item some line segments
in $T_G$ can guard $R_p$,
and that \item no such line segments are in $S$ since $R_p$ is unguarded. \end{inparaenum} 
Consider the maximal regions in $P$ that lie immediately above, below, left, and right of $R_p$;
any sliding camera that sees any part of $R_p$ must intersect one of these regions.
See Figure~\ref{fig:growingRp}(a) for an example; note that the hatched region
cannot contain any line segment in $S$ since $R_p$ is unguarded. Moreover, the hatched region must contain at least one line segment
of $T_G$ in both horizontal and vertical directions and without loss of generality we can assume that the length of these
line segments is maximal (i.e., both endpoints are on the boundary of $P$).

The rectangle $R_p$ defines a partition of $P$ into three parts: the \emph{vertical slab} through $R_p$, (i.e., the slab
whose sides are aligned with the vertical sides of $R_p$), the subpolygon of $P$ to the left of the vertical slab and the subpolygon of $P$ to the
right of the vertical slab. Similarly, another partition of $P$ can be obtained by considering the \emph{horizontal slab} through $R_p$; see
Figure~\ref{fig:growingRp}(b) for an example. We know that the union of the visibility regions of the line segments in $S$ is a
connected subregion of the plane. Therefore, the set $S$ can only be found on one side of each of the partitions of $P$ and so $S$ must be
in one \emph{corner} of the partitioned polygon. Without loss of generality, assume that $S$ is on the bottom left corner of the partitioned polygon
(see Figure~\ref{fig:growingRp}(b)).

Let $\overline{S}\subseteq T_G$ be the set of line segments that can see $R_p$. Note that $\overline{S}$
is non-empty since $H_u$ or $V_u$ sees $P$ entirely where $u$ denotes the nearest reflex vertex to $R_p$ in the horizontal or vertical slabs. Therefore, polygon $P$ is partitioned into three subpolygons (see Figure~\ref{fig:growingRp}(c)): the lower-left corner that is the location of $S$ denoted by $P_S$, the lower-right and upper-left corners that correspond to line segments in $\overline{S}$ denoted by $P_{\overline{S}}$, and the upper-right corner of the polygon that is unguarded denoted by $P_U$.
Each line segment in $\overline{S}$ intersects at least one line segment in $S$
since the line segments in $\overline{S}$ are not in $S$ and $S$ is feasible solution for the MMGG problem (see Figure~\ref{fig:growingRp}(c)).

\begin{figure}[t]
\centering%
\includegraphics[width=0.95\textwidth]{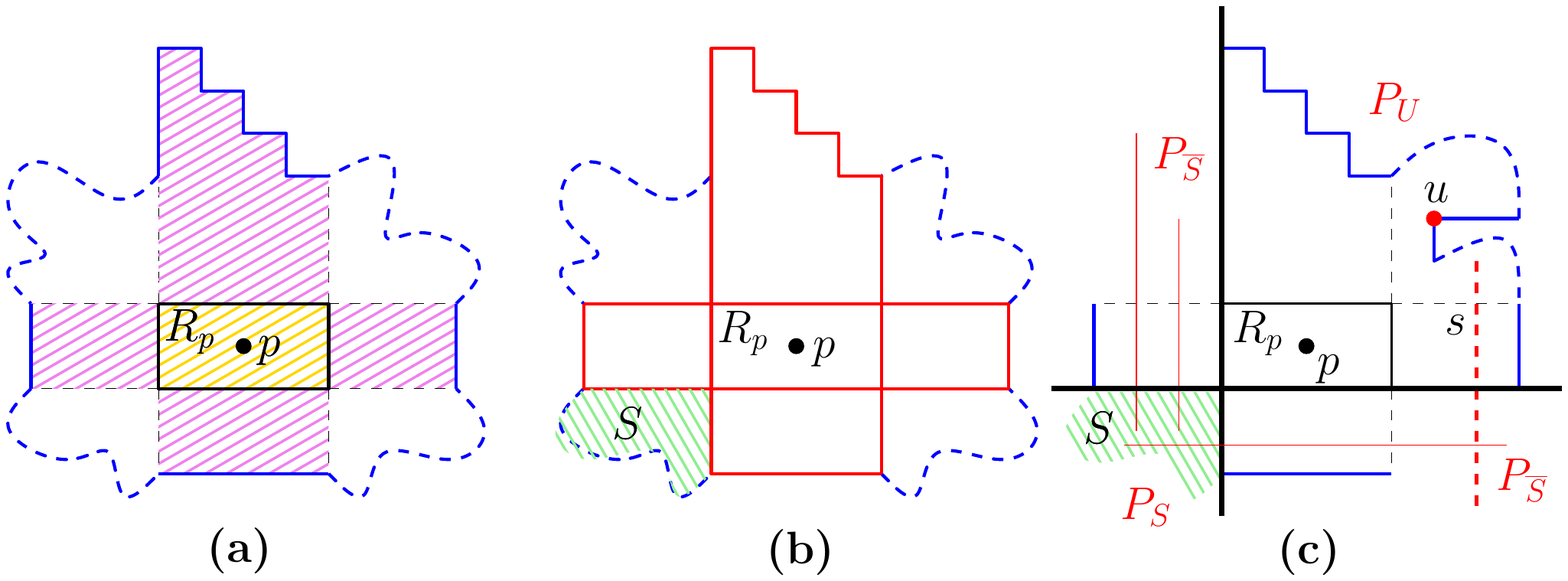}
\caption{(a) Point $p$ inside the polygon $P$ with rectangle $R_p$ hatched in gold. The purple hatched region
indicates the subregion of $P$ covered by growing $R_p$ orthogonally towards the boundaries of $P$; the boundary
of $P$ is shown in blue.
(b) The horizontal rectangle and the vertical histogram shown in red indicate, respectively, the horizontal and vertical slabs of the partitions induced
by rectangle $R_p$. The hatched region of $P$ on the bottom left corner of the partition indicates the location of the set $S$. (c) The partition of $P$ into three subpolygons $P_S$, $P_{\overline{S}}$ and $P_U$. The line segments in $\overline{S}$
(i.e., the set of line segments that can see $R_p$) are shown in red; observe that each line segment in $\overline{S}$ intersects at least one line
segment in $S$. The line segment $s$ illustrates Lemma~\ref{lem:opposite} and vertex $u$ is in support of Lemma~\ref{lem:contradiction}.}
\label{fig:growingRp}%
\end{figure}

\begin{lemma}
\label{lem:opposite}
No line segment in $T_G$ that is orthogonal to a line segment in $\overline{S}$ can intersect $P_U$.
\end{lemma}
\begin{proof}
To derive a contradiction, suppose without loss of generality that there is one such vertical line segment $s$ intersecting $P_U$ 
(as shown in Figure~\ref{fig:growingRp}(c)).
Since $S$ is connected and is a feasible solution for the MMGG problem, there must be a line segment in $S$ that
guards $s$. So, $S$ must contain a line segment in $P_{\overline{S}}$, which is a contradiction.
\end{proof}

By Lemma~\ref{lem:opposite}, we conclude that $R_p$ is guarded by the line
segments in $\overline{S}$,
but not by any line segment in $S$, and, furthermore, $S$ is restricted to a corner as described above (i.e., the subpolygon $P_S$). We now show that each of the regions of $P$ that are not guarded by
the line segments in $S$ must be a staircase with the reflex vertices oriented towards $S$.

To derive a contradiction, suppose that there exists a reflex vertex $u$ (as shown in Figure~\ref{fig:growingRp}(c)) in the unguarded subpolygon $P_U$ of $P$. However, no 
line segment in $T_G$ can intersect $P_U$ by Lemma \ref{lem:opposite}. This
contradicts the existence of $u$. Therefore, any unguarded region of $P$ by $S$ must be bounded by the line segments 
in $L(P)$\footnote{If the bounding line segment is not in $T_G$ then a dominating line segment must be in $T_G$. 
See Figure \ref{fig:ex_poly} for an example.} on adjacent horizontal and
vertical sides, and by a staircase of $P$ on the other sides; we call these regions the \emph{critical regions} of $P$, and denote $R_C$ to be the set of critical regions
of $P$. We now have the following lemma.


\begin{lemma}
\label{lem:contradiction}
Every point of $P$ that is not inside a critical region of $P$ is visible to at least one
line segment in $S$. Moreover, each critical region of $P$ is a staircase.
\end{lemma}

Let $OPT_{GG}$ denote an optimal solution for the MMGG problem on $G_P$. We first prove
that $\lvert OPT_{GG}\rvert\leq \lvert OPT_{GP}\rvert$.

\begin{lemma}
\label{lem:polygonAndGP}%
For any feasible solution $M$ for the MGSC problem on $P$, there exists a feasible solution $S'$ for the MMGG problem on $G_P$ such that $\lvert S'\rvert\leq\lvert M\rvert$.
\end{lemma}
\begin{proof}
Let $M$ be a feasible solution to the MGSC problem on $P$; that is, $M$ is a guarded set of orthogonal line segments
inside $P$ that collectively guard $P$. We construct a feasible solution $S'$ for the MMGG problem such that
$\lvert S'\rvert\leq\lvert M\rvert$. To compute $S'$, for each horizontal line segment $s\in M$ (resp., vertical
line segment $s\in M$), move $s$ up or down (resp., to the left or to the right) until it is collinear with a line segment
$s\in L(P)$. If $s\in T_G$, then add $s$ to $S'$; otherwise, add $s'$ to $S'$, where $s'\in L(P)$ is the line segment that dominates $s$. Note that there exists at least one such line segment $s'$ because otherwise the line segment $s$ would have not been removed from $L(P)$. It is straightforward to see that the union of visibility regions of line segments in $M$ is a subset of the visibility regions of line segments in $S'$. Since the camera travelling along each line segment in $M$ is seen by at least one other camera (see the definition of a guarded set of sliding cameras) and the grid $G_P$
is entirely contained in $P$, we conclude that $S'$ is a feasible solution for the MMGG problem on $G_P$. The inequality
$\lvert S'\rvert\leq\lvert M\rvert$ follows from the fact that each line segment in $M$ corresponds to at most one line
segment in $S'$. This completes the proof of the lemma.
\end{proof}

Next we need to find a set of minimum cardinality of orthogonal line
segments inside $P$ that collectively guard the critical regions of $P$.



\subsection{Guarding Critical Regions: A \boldmath$(3/2)$-Approximation Algorithm}
\label{subsec:exactAlgorithmForCriticals}
In this section, we give an approximation algorithm for the problem of guarding the critical regions of
$P$. The algorithm relies on reducing the problem to the minimum edge cover problem in graphs. The
minimum edge cover problem in graphs is solvable in $O(n^{5/2})$ time, where $n$ is
the number of graph vertices~\cite{vazirani1980}. We first need the following result:

\begin{lemma}
\label{lem:atLeastOneSegment}
Every critical region in $R_C$ is guarded entirely by some orthogonal line segment in $L(P)$.
\end{lemma}
\begin{proof}
Observe that if $P$ is a rectangle, then the MSC problem is trivial to solve. 
Suppose that $P$ is not a rectangle
and so it has at least one reflex vertex.
Furthermore, suppose that some regions of $P$ are not guarded by $S$ (the set of segments returned by
solving the MMGG problem on $G_P$), i.e., the set $R_C$ of critical regions of $P$ is non-empty.
Let $R\in R_C$ be a critical region of $P$. The lemma is implied by the fact that there exists at least one
reflex vertex on the boundary of $R$; this is because $P$ is not a rectangle and the set of line segments in
$S$ do not guard $R$. It is now
straightforward to see that one of the orthogonal line segments in $L(P)$ that passes through either the lowest or
the highest reflex vertex of $R$ can see the critical region $R$ entirely.
\end{proof}

Recall $R_C$, the set of critical regions of $P$. We construct a graph $H_P$ associated with $P$
as follows: for each critical region $R\in R_C$, we add a vertex $r_H$ to $H_P$. Two vertices
$r_H$ and $r'_H$ are adjacent in $H_P$ if and only if there exists an orthogonal line segment inside
$P$ that can guard both critical regions $R$ and $R'$ entirely. Finally, we add a self-loop edge for every
isolated vertex of $H_P$.


\begin{lemma}
\label{lem:atMostTwoCriticals}
Any orthogonal line segment inside $P$ can guard at most two critical regions of $P$ entirely.
\end{lemma}
\begin{proof}
Let $s$ be an orthogonal line segment inside $P$. Observe that the only way for $s$ to guard a critical region $R$ entirely is that at least
one of its endpoints lies on the boundary of $P$, covering one entire edge of $R$; see Figure~\ref{fig:atMostTwoCriticals} for an example.
Therefore, $s$ can guard at most two critical regions of $P$.
\end{proof}


\begin{figure}[t]
\centering%
\includegraphics[width=0.3\textwidth]{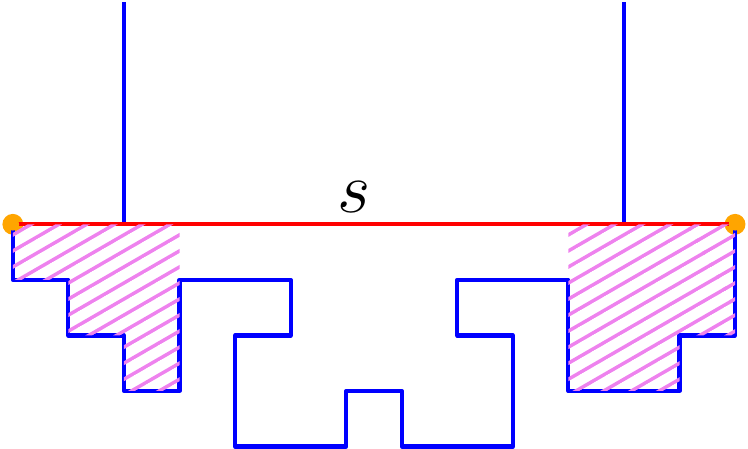}
\caption{Any orthogonal line segment inside $P$ can guard at most two critical regions
of $P$ entirely.}
\label{fig:atMostTwoCriticals}%
\end{figure}

\begin{lemma}
\label{lem:criticalEdgeCoverEquivalence}
The problem of guarding the critical regions of $P$ using only those line segments that may individually guard a critical region reduces to the minimum edge cover problem on $H_P$.
\end{lemma}
\begin{proof}
We prove that \begin{inparaenum}[(i)]\item for any solution $S$ to the minimum edge cover
problem on $H_P$, there exists a solution $S'$ for guarding the critical regions of $P$ such
that $\lvert S'\rvert=\lvert S\rvert$, and that \item for any solution $S'$ to the problem of guarding
the critical regions of $P$, there exists a solution $S$ to the minimum edge cover problem on
$H_P$ such that $\lvert S\rvert=\lvert S'\rvert$. \end{inparaenum}\\

\noindent{\bf Part 1.} Choose any edge cover $S$ of $H_P$. We construct a solution $S'$ for guarding the critical regions of $P$ as follows. For each edge
$e=(r_H, r'_H)\in S$ let $s_e$ be the line segment in $P$ that can see both critical
regions $R$ and $R'$ of $P$; we add $s_e$ to $S'$. It is straightforward to see that
the line segments in $S'$ collectively guard all critical regions of $P$.\\

\noindent{\bf Part 2.} Choose any solution $S'$ for guarding the critical regions of $P$. We
now construct a solution $S$ for the minimum edge cover problem on $H_P$. By Lemma~\ref{lem:atMostTwoCriticals},
we know that every line segment in $S'$ can see at most two critical regions of $P$.
First, for each line segment in $S'$ that can see exactly one critical region $R$ of $P$, we add the
self-loop edge of $H_P$ that corresponds to $R$ in $S$. Next, for each line segment $s\in S'$ that
can see two critical regions of $P$, we add to $S'$ the edge in $H_P$ that corresponds to $s$.
Since any line segment in $S'$ can see at most two critical regions of $P$, we conclude that every
vertex of $H_P$ is incident to at least one edge in $S$ and, therefore, $S$ is a feasible solution for
the minimum edge cover problem on $H_P$.
\end{proof}

In general, it is possible for the solution $S'$ to be non-optimal. Only those edges which may individually guard a critical region were considered, while an optimal guarding solution may use two line segments to collectively guard a critical region, as shown in Figure \ref{fig:ex_poly}. By Lemma \ref{lem:atLeastOneSegment}, $S'$ requires at most one edge for each critical region and,
therefore, the number of guards returned y our algorithm is at most equal to the number of critical regions.
If an optimal solution uses two segments to collectively guard one critical region, 
these two edges suffice to guard three critical regions, while our solution uses three segments to guard the same three
critical regions. This results in an approximation factor of $3/2$ in the number of segments used to guard 
the set of critical regions.

\begin{figure}[t!]
\centering%
\includegraphics[width=0.9\textwidth]{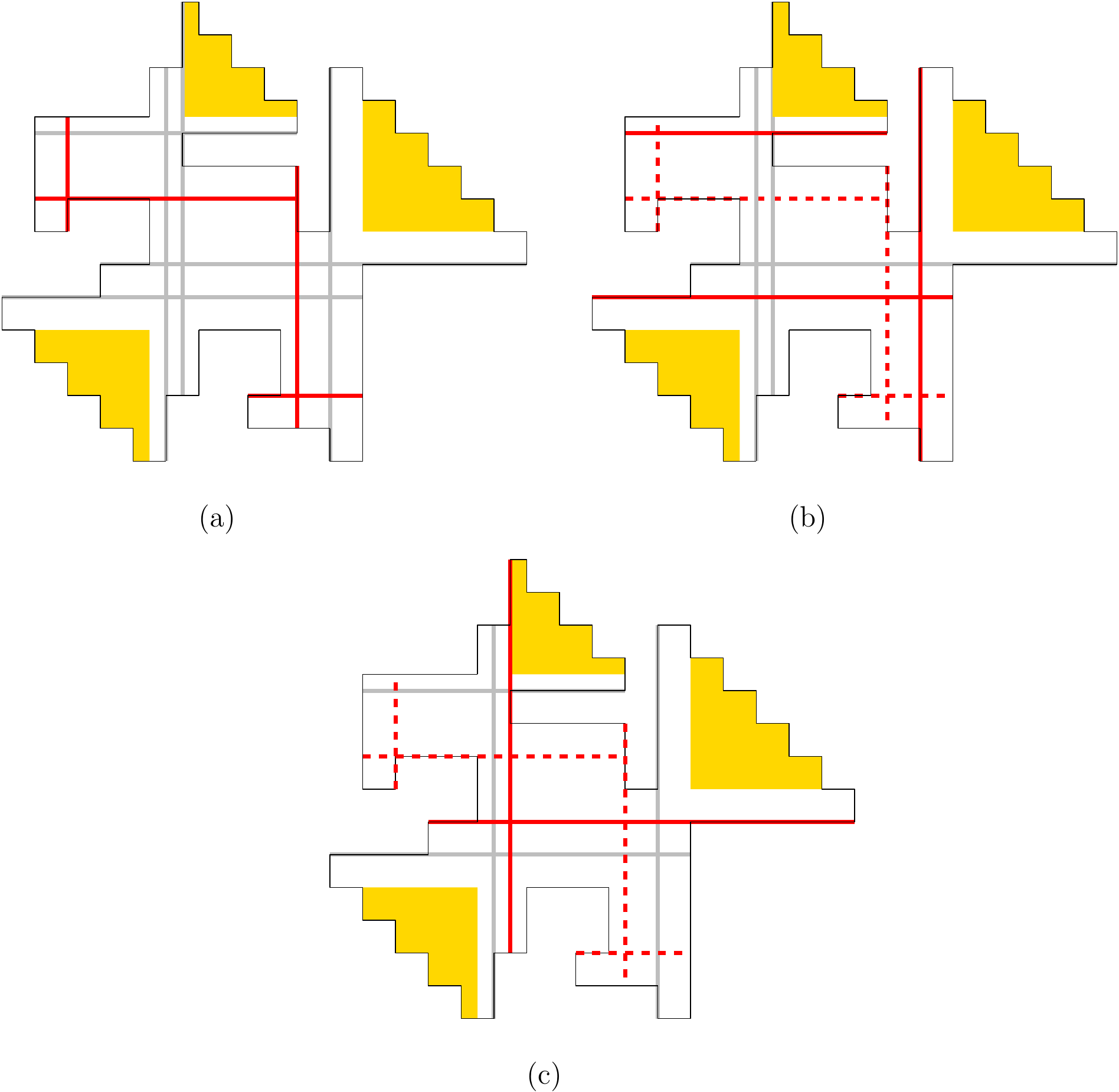}
\caption{An example of a polygon for which two line segments may collectively guard a critical region.  (a) The polygon has unguarded critical regions after finding an optimal MMGG solution. The thick red line segments indicate guarding lines, light gray line segments are unused, and the shaded regions are unguarded.
(b) A possible solution (solid red line segments) for guarding the critical regions using the edge guarding approach. (c) An optimal critical region guarding solution for the polygon.}
\label{fig:ex_poly}%
\end{figure}

We now examine the running time of the algorithm. Let $n$ denote the number of vertices of $P$.
To compute the critical regions of $P$, we first compute the set of staircases of $P$ and, for each of them, we check to see whether
they are guarded by the set of line segments in $S$. Each critical region of $P$ can be found
in $O(n)$ time and so the set of unguarded critical regions of $P$ is easily computed in $O(n^2)$ time.
Moreover, the graph $H_P$ can be constructed
in $O(n^2)$ time by checking whether there is an edge between every pair of vertices of the graph. Therefore,
by Lemma~\ref{lem:criticalEdgeCoverEquivalence} and the fact that the minimum edge cover
problem is solvable on a graph with $n$ vertices in $O(n^{5/2})$ time, we have the following
lemma.
\begin{lemma}
\label{lem:exactAlgForCriticals}
There exists a $(3/2)$-approximation algorithm for solving the problem of guarding the critical regions of
a simple orthogonal polygon $P$ with $n$ vertices in $O(n^{5/2})$ time.
\end{lemma}

Given any simple orthogonal polygon $P$, we find a set of sliding cameras that guards $P$ 
by first solving the instance $G_P$ of the MMGG problem determined by $P$. This may leave a set of critical regions within $P$
that remain unguarded. We then add a second set of sliding cameras to guard these critical regions.
Recall the set $S$, an optimal solution to the MMGG problem on $G_P$, found in $O(n^2)$ time,
where $n$ is the number of vertices of $P$.
By Lemma~\ref{lem:exactAlgForCriticals}, we approximate the problem of guarding the critical
regions of $P$ in $O(n^{5/2})$ time;
let $S_C$ be the solution returned by the algorithm. 
Since the union of the critical regions of $P$ is
a subset of $P$, any feasible solution to the MSC problem also guards the critical regions of $P$. 
Therefore, $\lvert S_C\rvert\leq (3/2) \cdot \lvert OPT_P\rvert$.
Moreover, by Lemma~\ref{lem:polygonAndGP} we know $\lvert OPT_{GG}\rvert\leq \lvert OPT_{GP}\rvert$ and since
$\lvert OPT_{GP}\rvert\leq 2\cdot\lvert OPT_{P} \rvert$ we have that $\lvert OPT_{GG}\rvert\leq 2\cdot\lvert OPT_P\rvert$.
Therefore, by combining $S$ and $S_C$ we obtain a feasible solution
to the MSC problem whose cardinality is at most 7/2 times $\lvert OPT_P\rvert$; that is,
$\lvert S\cup S_C\rvert\leq (7/2) \cdot\lvert OPT_P\rvert$. This gives the main result:


\begin{theorem}
\label{thm:msc3Approx}
There exists an $O(n^{5/2})$-time $(7/2)$-approximation algorithm for the MSC problem on any simple orthogonal polygon $P$ with $n$ vertices.
\end{theorem}

As a consequence of our main result, we note that $\lvert S_C\rvert\leq (3/2) \cdot \lvert OPT_P\rvert \leq (3/2) \cdot\lvert OPT_{GP}\rvert$ and again by Lemma~\ref{lem:polygonAndGP} we know $\lvert OPT_{GG}\rvert\leq \lvert OPT_{GP}\rvert$. Therefore, $\lvert S\cup S_C\rvert\leq (5/2) \cdot\lvert OPT_{GP}\rvert$. To show that the set $S\cup S_C$ is a feasible solution for the MGSC problem, we first observe that every line segment in $S$ is guarded by at least one other line segment in $S$. Moreover, every grid segment that is not in $S$ is guarded by some line segment in $S$ because $S$ is a feasible solution for the MMGG problem. This means that every line segment in $S_C$ is guarded by at least one line segment in $S$. Therefore, we have the following result.
\begin{corollary}
Given a simple orthogonal polygon $P$ with $n$ vertices, there exists an $O(n^{5/2})$-time $(5/2)$-approximation algorithm for the MGSC problem on $P$.
\end{corollary}

\section{Conclusion}
\label{sec:conclusion}
In this paper, we studied a variant of the art gallery problem, introduced by Katz and Morgenstern~\cite{katz2011},
where sliding cameras are used to guard an orthogonal polygon and the objective is to guard the polygon with
minimum number of sliding cameras. We gave an $O(n^{5/2})$-time $(7/2)$-approximation algorithm to this problem
by deriving a connection between a guarded variant of this problem (i.e., the MGSC problem) and the problem of guarding simple grids with sliding cameras. The complexity of the
problem remains open for simple orthogonal polygons. Giving an $\alpha$-approximation algorithm, for any
$\alpha< 7/2$, is another direction for future work. Finally, studying the MGSC problem (the complexity of the problem or improved approximation results) might be of independent interest.\\

\subsection*{Acknowledgements} The authors thank Mark de Berg and Matya Katz for
insightful discussions of the sliding cameras problem.

\bibliographystyle{plain}
\bibliography{ref}

\end{document}